\documentclass[manyauthors]{fundam} %

\usepackage{graphicx,xcolor}
 \usepackage{multirow}
\usepackage{amssymb}
\usepackage{graphicx,amsmath}
\usepackage{algorithmicx}
\usepackage{algorithm}
\usepackage[noend]{algpseudocode}
\usepackage{cleveref}
\usepackage{geometry}
\usepackage{placeins}
\usepackage{complexity}
\usepackage{url}

\newcommand{\aref}[1]{Alg.~\ref{#1}}

 \usepackage{hyperref}

\begin{document}

\setcounter{page}{205}
\publyear{2021}
\papernumber{2097}
\volume{184}
\issue{3}

   \finalVersionForARXIV

\title{Clustering Geometrically-Modeled Points in the Aggregated Uncertainty Model}

\author{Vahideh Keikha\thanks{Address for correspondence:  The Czech Academy of Sciences, Institute of Computer Science,
                 Pod Vod\'{a}renskou v\v{e}\v{z}\'{\i} 2, 182 07 Prague, Czech Republic.}\thanks{Also affiliated at: Department
                 of Computer Science, University of Sistan and Baluchestan, Zahedan, Iran. \newline \newline
              \vspace*{-6mm}{\scriptsize{Received  November 2021; \ accepted January 2022.}}}
          \\
The Czech Academy of Sciences\\
 Institute of Computer Science\\
Prague, Czech Republic\\
keikha@cs.cas.cz
\and
 Sepideh Aghamolaei\\
 Department of Computer Engineering\\
 Sharif University of Technology\\
  Tehran, Iran\\
 aghamolaei@ce.sharif.edu
 \and  Ali Mohades\\
 Department of Mathematics and
 Computer Sci.\\
  Amirkabir University of Technology\\
   Tehran, Iran\\
 mohades@aut.ac.ir
 \and
 Mohammad Ghodsi\\
Department of Computer Engineering\\
 Sharif University of Technology\\
  Tehran, Iran\\
		ghodsi@sharif.edu
}

\runninghead{V. Keikha et al.}{Clustering Geometrically-Modeled Points in the Aggregated Uncertainty Model}

 \maketitle

\vspace*{-8mm}
\begin{abstract}
  The $k$-center problem is to choose a subset of size $k$ from a set of $n$ points such that the maximum distance from each point to its nearest center is minimized.
  Let $Q=\{Q_1,\ldots,Q_n\}$ be a set of polygons or segments in the region-based uncertainty model, in which each $Q_i$ is an uncertain point,
  where the exact locations of the points in $Q_i$ are unknown.
  The geometric objects such as segments and polygons can be models of a point set.
  We define the uncertain version of the $k$-center problem as a generalization in which
  the objective is to find $k$ points from $Q$ to cover the remaining regions of $Q$ with minimum or maximum radius of the cluster to cover at least one or all exact instances of each $Q_i$, respectively. %
  We modify the region-based model to allow multiple points to be chosen from a region, and call the resulting model the {\em aggregated uncertainty model}.

  All these problems %
  contain the point version as a special case, so they are all NP-hard with a lower bound 1.822 for the approximation factor. We give approximation algorithms for uncertain $k$-center of a set of segments
  and polygons. We also have implemented some of our algorithms on a data-set to show our theoretical performance guarantees can be achieved in practice.

\medskip\noindent
\textbf{Keywords:}
$k$-center \and Uncertain data \and Approximation algorithms
\end{abstract}

\section{Introduction}
$k$-center is a classic problem in the fields of computational geometry, data mining, and approximation algorithms.
Suppose a set of $n$ points is given.
 The goal of metric $k$-center is to choose a subset of size $k$ called $C$ from a set of $n$ points such that the maximum distance from each point to its nearest center is minimized.

Metric $k$-center is NP-hard and it has a tight $2$-approximation algorithms in metric spaces \cite{vazirani2013approximation}.
According to Theorem 2.1 in~\cite{feder1988optimal}, the lower bound on the approximation factor of $k$-center in the Euclidean plane is $1.822$.

The data {\em uncertainty} usually comes from the error in the precision of representing numbers~\cite{salesin1989epsilon}, the error in the measurement of input data such as GPS data~\cite{loffler2009data} or other sources. Sometimes the probability of data is known based on previous data or measurement error~\cite{cormode2008approximation,suri2013most1}.
Uncertainty sometimes affects the correctness of geometric algorithms, for example, a convex hull might not be convex at all as a result of reduced precision. This has motivated a lot of uncertainty models over the years.

We introduce a new model for uncertainty, which we call the ``aggregated
uncertainty model". This allows the multiplicity of the realizations of an uncertain point to exceed one, which is a realistic assumption for many applications, including anonymized data, continuous data such as GPS data or sensor outputs, and repeated data which were summarized in the geometric version of the input as a single point. In this paper, when we refer to uncertainty, it is the aggregated region-based uncertainty model, which is the aggregated uncertainty model applied to a set of points in the region-based uncertainty model (we have defined it later in this section).

\medskip
Through the paper, the { covering} means the distance between the covered point and its closest center is at most the radius of the covering, and the goal is to minimize the radius of covering. We introduce and study the following two problems: (see \Cref{sec:def} for formal definitions):
\begin{enumerate}
\itemsep=0.9pt
	\item \textbf{MinMax segments clustering}: covering a set of $n$ segments with $k$ segments as the centers, and the maximum distance from a point on a segment to its center is minimized. And,
	
	\item \textbf{Domain-restricted $k$-center of polygons}: covering the points of a set of polygons by a set of $k$ points inside those polygons, and the maximum distance from a point on a polygon to its center is minimized.

\end{enumerate}

For each problem, we discuss two cases: the maximum version and the minimum version of the problem. In the maximum version, the goal is to cover the whole input, while in the minimum version the goal is to hit each input, i.e. to cover at least one point of each input.

Many variations of $k$-center have been studied before, in the following, we review the ones most related to our problems as explained below.

\paragraph{$k$-center Variations} In {\em $k$-line center} problem~\cite{agarwal2005approximation}, the input is a set of lines and the goal is to find $k$ lines (or cylinders in higher dimensions) as centers that minimize the distance from each point to the nearest line. In the same paper, a $(1+\epsilon)$-approximation algorithm for this problem in $O(n\log n)$ time was given.

Another related problem is stabbing a set of line segments with two congruent squares of minimum size\cite{sadhu2019linear}. The authors introduced a $\sqrt{2}$-approximation algorithm with running time $O(n)$. They have shown that their algorithm can also be applied to find two congruent disks of minimum size to cover the line segments. Note that this is different from the MinMax segments clustering since our objective is to find $k$ segments from the input as centers, such that this selection minimizes the radius of the covering.

A constrained version of our second problem is covering the area of a convex polygon $P$ of $n$ vertices by $k$ centers on its boundary so that each circle has the smallest possible radius~\cite{du2014approximation}. The authors of~\cite{du2014approximation} provide a 1.8841-approximation $O(nk)$ time algorithm, that first computes an approximate smallest bounding box $R$ for $P$, and solves the problem for $R$. Then, they translate the centers on the boundary of $R$, to the points on the boundary of $P$.
Later, a $(1+\frac 7 k+O(\epsilon))$-factor approximation algorithm for any $\epsilon>0$ and $k \ge 7$ was given, which runs in $O(n^2 (\log r_{opt})+\log \frac 1 \epsilon)$ time~\cite{basappa2015constrained}.

\paragraph {Uncertainty Models}
Several formulations have been proposed for modeling data uncertainty over the years, including epsilon  geometry~\cite{salesin1989epsilon}, the probabilistic model~\cite{cormode2008approximation,suri2013most1}, the region-based model~\cite{loffler2009data}, and the domain-restricted models~\cite{edalat2001convex}.

In the region-based model (which is equivalent to the locational model when each uncertain point has a uniform distribution over its region), two types of regions where each point can occur are usually discussed: the continuous and the discrete region.
The problem asks for the minimum and maximum solutions, defined as the solutions that work in the luckiest (best possible distribution based on the objective function) and unluckiest (worst possible distribution based on the objective function) of input points.

Several problems are studied in the region-based uncertainty model, where the objective was minimizing or maximizing the size of the area/perimeter of the convex hull ~\cite{loffler2010largest}, the area of the bounding box, the length of width and diameter, and the area of the smallest enclosing disk~\cite{loffler2010largest2}, or the area of the inscribing convex $k$-gon~\cite{keikha2017largest}, etc.
The smallest area/perimeter convex hull that contains at least one point of each imprecise point also called {\em polygon transversal}. For a given set of line segments, the smallest area/perimeter convex polygon that intersects all the segments can be computed in $O(n \log n)$ time~\cite{mukhopadhyay2008intersecting,rappaport1995minimum,boissonnat1996convex}.

\paragraph{Uncertain	$k$-center} In the discrete model of uncertainty in which each uncertain point is modeled by a discrete set of points with an assigned probability of occurrence, the $k$-center problem was studied in~\cite{cormode2008approximation,huang2017stochastic,wang2015one}, where, in all of them, the objective is to minimize the maximum expected distance from the uncertain points.
Also, the authors of~\cite{cormode2008approximation,wang2015one} only consider the problem in 1D space.
The $1$-center of a set of uncertain points also studied for rectilinear distances~\cite{wang2018computing}, where the uncertainty of each point is modeled as a set of discrete points with assigned probability, but the objective is still minimizing the maximum expected rectilinear distance to the uncertain points.

In the following, we briefly compare the introduced uncertainty model with the existing models.
In the region-based model, the objective is to compute the maximum and minimum possible feasible solutions, with the existing error bound estimations. In the domain-restricted model, the objective is to just find a class of feasible solutions. But, in the aggregated uncertainty model, the objective is to compute a compact subset of the input which gives a guaranteed approximation for any exact instance of the data, assuming the error bound estimation is given.

\paragraph{Aggregated Uncertainty Model}
A common event in modeling uncertain data is when multiple points are mapped to a single uncertain point. However, most existing uncertainty models allow only a single point to be chosen from each uncertainty region. We remove this assumption from region-based uncertainty models in our model, which is described in \Cref{def:uncertain}.
Also, in the current study, by an uncertain point we mean a point that is specified
by a region in which the point may lie.

\begin{definition}[Aggregated Uncertainty Model (AGU)]\label{def:uncertain}
For a set of uncertain points represented with regions $Q_1,\ldots,Q_n$ in $\mathbb{R}^2$ as inputs to a problem $\mathcal{P}$ in the AGU model, a feasible solution for $\mathcal{P}$ is a solution which is feasible for all subsets of points in $\cup_{i=1}^n Q_i$, such that at least one point from each set $Q_i$ is chosen. Indeed each uncertain point in the AGU model is defined by a continuous region, and corresponds to at least one point in the real world.

	An optimization problem in this model is solved by computing a minimum-cost solution with at least one point in each region and a maximum-cost solution with the same constraints. i.e, the optimal solution of the minimization version of a problem $\mathcal{P}$ in the AGU model is a minimum of the feasible solutions.
\end{definition}

\paragraph{$k$-Center in the aggregated uncertainty model.}
Since we allow multiple points to be chosen from each uncertainty region, the uncertainty models of the input points and the chosen centers can be different. In the domain-restricted $k$-center of polygons problem that we discussed in this paper, we assume the minimum cost for the set of centers is intended. More specifically, we allow each center to be an arbitrary point in an input polygon (uncertainty region).
Both the minimum and the maximum models for the output of an uncertainty problem, as defined for the region-based uncertainty models, have been discussed for the input points.

\paragraph{Motivation}
As an application, suppose for a huge set of $N$ data points we are allowed to draw a set of $n\ll N$ regression lines. Then each line segment in the line segment uncertainty model coincides with one of these regression lines (which we have bounded the length of each line by the range of the corresponding input points), where we require to classify all the points based on the computed regression lines.

In the domain-restricted $k$-center of polygons, we assume each polygon represents a geometrical label, e.g. areas on a map. Conversely, we may have a set of labeled points, and for the points with equal labels, we have computed the corresponding region which here the region is modeled by a polygon.
 Then the objective is to make a classification for such polygons, which is persistent for any point which lies in the polygons in the future.

\subsection*{Contributions}
{
We introduce a problem which we call multi-interval set cover, which appears as a sub-problem in many problems, e.g., the wireless network problems in which each interval can be seen as a moving sensor 	with bounded range~\cite{khelifa2009monitoring,kloder2007barrier}, or admissible choices in a game environment~\cite{bopardikar2008discrete}.}
\begin{itemize}
\itemsep=0.95pt
	\item
	In the MinMax segment clustering the objective is that of choosing $k$ segments as centers to cover the remaining regions modeled by a set of segments.
	We study both the maximization and the minimization versions that give some bounds on the radius of the clusters.
	 We show that $k$-center of segments is NP-hard by an approximation-preserving reduction from the set cover problem, so it cannot be approximated by a factor better than $\Omega(\log n)$ in polynomial time.
 However, the approximation factor of the multi-interval set cover affects only the number of clusters, which results in a bicriteria approximation for the problem.
A summary of the results is given in Table~1.

\renewcommand{\arraystretch}{1.3}
\begin{table}[h!]
	\centering
    \caption{The summary of the results. In our bi-criteria approximation for clustering, we assume $\beta k$ centers are used and the radius
        is $\alpha r$. $n$  is the number of input objects. The algorithms that assume the aspect ratio is fixed are marked with
        $^\ddagger$.} \label{table:results}
	\scalebox{0.87}{
   \begin{tabular}{|l|l|l|l|l|l|l|}
		\hline
		Problem & Region & Type & $\beta$ & $\alpha$ & Time & Refs. \!\!\\
		\hline
		\hline
		MinMax segment clustering ($k=1$)\! & Segments & Max & 1 & 1 & $O(n^2)$ & \aref{alg:1center}\!\\
		MinMax segment clustering\! & Segments & Max & $O(\log n)$ & $1+\epsilon$ & $O(\frac{n^3}{\epsilon^4})$ & \aref{alg:new}$^\ddagger$\!\\
		MinMax segment clustering\! & Segments & Max & $1$ & $\geq 1.822$ & $O(poly(n))$ & \cite{feder1988optimal}\!\\
		\hline
		MinMax segment clustering ($k=1$)\! & Segments & Min & 1 & 1 & $O(n^2)$ & \aref{alg:m1center}\!\\
		MinMax segment clustering\! & Segments & Min & $O(\log n)$ & $1+\epsilon$ & $O(\frac{n^3}{\epsilon^4})$ & \aref{alg:new2}\!\\
		\hline
	\end{tabular} }
	\end{table}

\item We prove a special case of the geometric set cover problem, which we call the multi-interval set cover problem is NP-hard, and prove it is NP-hard to approximate it by a factor better than $\Omega(\log n)$.

\item In the domain-restricted $k$-center of polygons in the AGU model, we solve the $k$-center of a set of polygons which is that of choosing $k$ points from the input as the centers to cover the area of the input polygons.
We study both the maximization and the minimization versions that give some bounds on the radius of the clusters over all possible distributions of exact instances. We give constant factor approximation algorithms for this problem. The results are summarized in \Cref{table:results2}.

\renewcommand{\arraystretch}{1.3}
\begin{table}[h]
	\centering
\caption{The summary of the results. In our bi-criteria approximation for clustering, we assume $\beta k$ centers are used and the radius is $\alpha r$. $n$ is the number of input objects. In polygon clustering problems, $N$ is the total complexity of the input polygons.
		The algorithms that assume the aspect ratio is fixed are marked with $^\ddagger$.}	\label{table:results2}
\scalebox{0.9}{
 	\begin{tabular}{|l|l|l|l|l|l|l|}
		\hline
		Problem & Region & Type & $\beta$ & $\alpha$ & Time & Refs.\\
		\hline
		\hline
		$k$-Center & Convex Polygons & Max & 1 & $2+\epsilon$ & $O(N+\frac{1}{\epsilon^2})$ & \aref{alg:kcluster}$^\ddagger$\\
		$k$-Center & Polygons & Max & $1$ & $1+\frac{\sqrt{3}}{2}$ & $O(kN)$ & \aref{alg:kclusterarbitrary}\\
		$k$-Center & Polygons & Max & $1$ & $\geq 1.822$ & $O(poly(n))$ & \cite{feder1988optimal}\\
		\hline
		$k$-Center & Polygons & Min & $O(1)$ & $1$ & $n^{O(1)}+O(\frac{1}{\epsilon^2})$ & \aref{alg:minkcenterC}$^\ddagger$\\
				\hline
	\end{tabular} }
	\end{table}

\item We have implemented our algorithms for solving (1) the maximum version of domain-restricted $k$-center of polygons of a data-set of size 4,491,143 of labeled points.
(2) the existing algorithm to solve this problem.
 We have not observed any significant loss of accuracy of the optimal solutions by our algorithms in practice.
\end{itemize}

\section{Preliminaries} \label{sec:preliminaries}

\paragraph{Bicriteria Approximation and Pseudo-Approximation}
In bicriteria approximation \cite{gonzalez2007handbook}, two objective functions are approximated simultaneously by the solution of the algorithm. The approximation factor is usually reported as a pair $(\alpha,\beta)$, where $\alpha$ is the approximation factor of the first objective function and $\beta$ is the approximation factor of the second objective function.

\paragraph{Smallest Enclosing Disk (SED)}
Given is a set $P$ of points in the plane. The aim is to place a point $c$ that
minimizes the maximum distance
to the points of $P$. Then $c$ determines the center of the {\em smallest enclosing disk} of $P$.
Megido~\cite{megiddo1983linear} showed that the smallest enclosing disk can be computed in linear time.

\paragraph{Uncertain Data}
In epsilon-geometry~\cite{salesin1989epsilon} the authors introduced a framework that provides algorithms to determine the perturbation $\epsilon$
based on the rounding errors in which the output remains correct.
If only the topology of the structure of the output needs to be kept, the value $\epsilon$ is called the {\em tolerance} of the structure~\cite{abellanas1999structural}.

In most probabilistic methods, a probability distribution is assigned to each uncertain point. These models are further divided into the {\em locational} models and the {\em existential} models. In the locational model, each uncertain point
always exists but its exact location is unknown. Then, the location of each uncertain point will be determined by a specific probability distribution, but in the existential model, each uncertain point has a
certain exact location but its existence is unknown. We refer the reader to~\cite{agarwal2018range,xue2019expected} and the references therein.

The domain-restricted uncertainty models~\cite{edalat2001convex,khanban2003computing,davari2019convex} focus on the solution set of the problem and therefore does not require an objective function, unlike the region-based models~\cite{loffler2010largest} whose results highly depend on the objective function. For example, the convex hull of the domain-restricted uncertainty model is defined by an interior region and an exterior region, where the convex hull includes the interior region and is contained in the exterior region~\cite{edalat2001convex}.

In the discrete model, each uncertain point is usually modeled by a discrete set of colored points, and the objective is selecting one point from each color such that a specific measure of points is minimized or maximized. There are several recent studies on this topic, see, e.g.,~\cite{kazemi2019approximability}.

\paragraph{Smallest Color-Spanning Disk}
Given a set of points with different colors, the goal is to find the smallest disk that contains at least one point from each color \cite{abellanas2001smallest}. There is a $O(\tau n \log n)$ time algorithm, for $\tau$ colors by computing the upper envelope of Voronoi surfaces \cite{huttenlocher1993upper,vsarir1995davenport}.

\paragraph{Colorful $k$-Center}
The colorful $k$-center is a generalization of $k$-center where each point has a color and there is a minimum requirement for each color that must be covered in the final solution. A $O(1)$-approximation $n^{O(c)}$ time algorithm for this problem with $c$ colors was given in~\cite{bandyapadhyay2019constant}.

\paragraph{$\epsilon$-Net}
For any subset $N\subset A$ and $0\leq \epsilon \leq 1$, and a shape $r$ from range space $R$, if the inequality $|r\cap A|\geq \epsilon |A|$ implies that $r$ contains at least one point of $N$, then $N$ is called an $\epsilon$-net~\cite{har2005geometric}.

When the range space is a set of disks, the number of independent samples for building an $\epsilon$-net is $O(\frac{1}{\epsilon}\log\frac{1}{\epsilon \delta})$, where $1-\delta$ is the probability of the resulting set to be an $\epsilon$-net and $\epsilon>0$ is a given constant (this is a special case of the theorem in~\cite{har2005geometric}).

For computing the set cover of a set of objects with a fixed VC dimension $d$, there exists an $O(d \log (nd) )$-approximation algorithm by using an $\epsilon$-net, where $n$ is the complexity of the optimal solution. However, for a set of
polygons of $n$ vertices, the VC dimension at least equals $2n+1$, which implies that the approximation factor of the set cover equals $O((2n+1)\log((2n+1)n)$ in the worst-case.

\paragraph{Geometric Set Cover and Geometric Hitting Set}
In the geometric set cover problem, a range space $\Sigma =(X,R)$ is given as the input, where $X$ is a set of points in $\mathbb {R} ^{d}$ and $R$ is a family of subsets of $X$.
The goal is to select a minimum-size subset $C\subseteq R$, such that any point in $X$ is inside a shape in $C$.

The geometric hitting set problem, finds the smallest subset $H\subseteq X$, such that any shape in $R$ contains at least one point of $H$.

By reduction from facility location, the problem was proved to be NP-hard for disks in 2D, and
$o(\log n)$-approximation algorithms for fat triangles, pseudo-disks, and other fat objects exist~\cite{clarkson2005improved}.
A randomized constant-factor approximation based on the concept of $\epsilon$-nets exists for the geometric hitting set and set cover of disks with $O(n~\polylog(n))$ time in $\mathbb{R}^2$ \cite{agarwal2014near,agarwal2020near}. The approximation factor of hitting set was later improved to $13.4/\epsilon$~\cite{bus2016tighter}. Several PTAS algorithms for hitting set~\cite{mustafa2010improved} and set cover~\cite{durocher2015duality} of pseudo-disks also exist.

\paragraph{Ply}
A set of $1$-dimensional intervals $S_i, i=1,\ldots,n$ are given, and $U=\cup_{i=1}^n S_i $ is the union of all input intervals (the universe).
The number of intervals $S_i,i=1,\ldots,n$ that contain a single point of the universe interval is called the ply of that universe~\cite{buchin2017folding}.

\section{Problem definition} \label{sec:def}
The $k$-center problem on points can be generalized to multiple versions on segments and polygons. In one of our generalizations, all the points of an object which is chosen as a center can cover other points. We study this model on a set of segments and call it the {\em MinMax Segment Clustering}. In the other version, only $k$ points chosen from input objects can act as centers and cover other points. We study this model on a set of polygons and call it {\em Domain-Restricted $k$-center of polygons}.

\subsection{MinMax segments clustering}%

\paragraph{\textbf{Maximum Cost}} We are given a set $S=\{s_1,\ldots,s_n\}$ of $n$ segments and we want to choose a subset $C$ of size $k$ of $S$ as the set of centers, such that the distance from any point on any segment in $S$ to its nearest center in $C$ is minimized.
In other words, the goal is to cover the points of all segments, using points from $k$ segments. An example of $1$-center of a set of segments is shown in \Cref{fig:segments_points}.

\medskip
Formally, the distance from a segment $s$ to the centers from $C$ is defined as follows:
\[
d(s,C) = \max_{p\in s} \min_{c\in C} d(p,c).
\]
The objective function of the $k$-center is therefore:
\[
\min_{\substack{C \subset S,\\|C|= k}} \max_{\substack{p \in s,\\ s_i \in S}} \min_{c \in C} d(p,c).
\]

Note that the distance defined over a set of segments (from other segments) is not a metric space.

\begin{figure}[h]
\vspace*{4mm}
	\centering
	\begin{minipage}{0.48\textwidth}%
		\centering
		\includegraphics[scale=1]{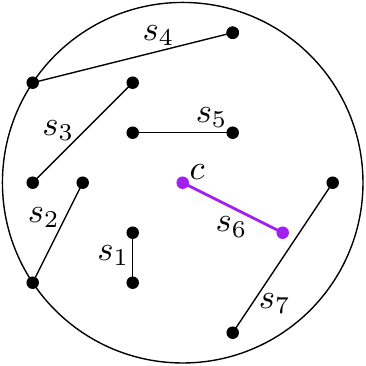}\vspace*{-1mm}%
		\caption{The maximum $1$-center of a set of segments with $C=\{s_6\}$. Note that  the distance from any point on any segment to $s_6$ is minimized. The point $c$ is the center of $1$-center.}
		\label{fig:segments_points}
	\end{minipage}%
	\hfill
	\begin{minipage}{0.48\textwidth}%
		\centering
		\includegraphics[scale=1]{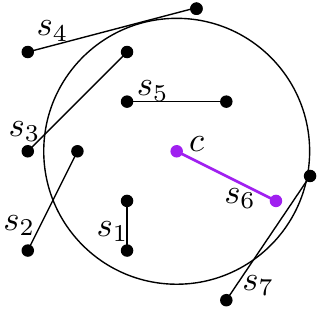}\vspace*{-1mm}
		\caption{The minimum $1$-center of a set of segments with $C=\{s_6\}$. Note that the distance from at least one point of each segment to $s_6$ is minimized. The point $c$ determines the center of the 1-center.}
		\label{fig:min_segments}
	\end{minipage}\vspace*{-4mm}
\end{figure}

\paragraph{\textbf{Minimum Cost}}
For a given set of segments, the goal of the {\em minimum $k$-center of segments} is to choose $k$ of these segments such that the distance from at least one point of each segment to its nearest center is at most $r$, and $r$ is minimized.
The objective function of the minimum $k$-center of a set $S$ of segments is as follows:
\[
\min_{\substack{C \subset S,\\|C|\leq k}} \max_{s \in S} \min_{p \in s} \min_{c \in C} d(p,c).
\]

See \Cref{fig:min_segments} for an example of $1$-center of a set of segments.

\subsection{Domain-restricted $k$-center of polygons}

\paragraph{\textbf{Maximum Cost}}
Let $Q=\{Q_1,\ldots,Q_n\}$ be a set of polygons.
The maximum $k$-center of polygons %
solves the $k$-center problem on a set $Q$ of polygonal shapes by finding a set $C$ of $k$ points inside the input polygons as centers such that the maximum distance from each point inside these polygons to its nearest center is minimized. An example of $2$-center of polygons is shown in \Cref{fig:polygons}.

\medskip
Formally, for a set $Q$ of polygons, the radius of clustering is defined as:
\[
\max_{Q_i\in Q} \max_{p\in Q_i} \min_{c\in C} d(c,p),
\]
and the objective function of the $k$-center of $Q$ is:
\[
\min_{\substack{C\subset (\cup_i Q_i),\\|C|= k}} \max_{\substack{p\in Q_i,\\Q_i\in Q}} \min_{c\in C} d(c,p).
\]
In the domain-restricted $k$-center of polygons, the set of centers must lie inside the polygons. %

\begin{figure}[h]
	\centering
	\begin{minipage}{0.45\textwidth}%
		\centering
		\includegraphics[scale=1]{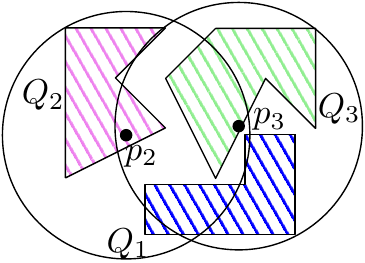}
		\caption{The maximum $2$-center of three polygons $Q_1,Q_2$ and $Q_3$ with $C=\{p_2,p_3\}$. The maximum distance from each point inside $Q_1,Q_2$ and $Q_3$ to its nearest center is minimized. }
		\label{fig:polygons}
	\end{minipage}%
	\hfill
	\begin{minipage}{0.45\textwidth}%
		\centering
		\includegraphics[scale=1]{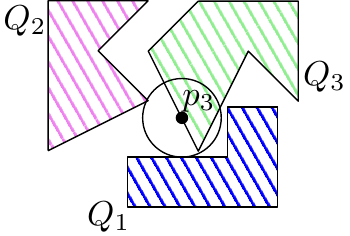}
		\caption{The minimum $1$-center of three polygons with $C=\{p_3\}$. The maximum distances from at least one point inside each of $Q_1,Q_2$ and $Q_3$ to the center $p_3$ are minimized.} %
		\label{fig:poly2}
	\end{minipage}\vspace*{-2mm}
\end{figure}

\paragraph{\textbf{Minimum Cost}}
If we require at least one point of each input polygon to be covered by $k$ centers chosen from points inside the input polygons, we call the problem {\em minimum $k$-center of polygons}. %
An example of $1$-center of polygons is shown in \Cref{fig:poly2}.

\medskip
Formally, the objective function of the minimum $k$-center of a set of polygons $Q$ is:
\[
\min_{\substack{C\subset (\cup_i Q_i),\\|C|= k}} \max_{Q_i\in Q} \min_{p\in Q_i} \min_{c\in C} d(c,p),
\]

\subsection{Multi-interval set cover}
We define a restricted version of the geometric set cover problem over a set of intervals, and call this special case {\em multi-interval set cover}. Formally, the problem is defined as follows:

\begin{definition}[Multi-Interval Set Cover]\label{def:multi}
A set of $n$ sets of intervals $Q_i, i=1,\ldots,n$ are given. The goal of multi-interval set cover is to find the minimum size subset of sets $Q_i, i=1,\ldots,n$, such that $\cup_{i=1}^n \cup_{s\in Q_i} s$ is covered.

\medskip
Two differences between multi-interval set cover problem and the well-known geometric set cover of intervals are first in the elements of the universe, which in the geometric set cover of intervals are points, and in the multi-interval set cover problem are intervals; and second, in the number of intervals in each set: each set in the multi-interval set cover is a set of intervals, while each set in the geometric set cover of intervals is a single interval (and therefore a continuous set of points).
\end{definition}
We introduce this problem and prove that it is NP-hard in \Cref{theorem:npc}. Also, we give approximation algorithms for this problem, using a reduction from set cover in \Cref{sec:intervalcover}.

\section{MinMax segments clustering}
Here, we discuss the $k$-center of a set of segments, where all the points on a set of segments must be within distance at most $r$ to their nearest point on a subset of size $k$ of the segments, such that $r$ is minimized.

\paragraph{The Hardness of MinMax Segment Clustering} %
Observe that the distance between the line segments is not symmetric as we have illustrated an example in \Cref{fig:nonsym}. Recall that we define the distance between a segment $s_i$ to $s_j$ as the distance between the furthest point of $s_i$ to the closest point of $s_j$.
Also, the triangle inequality no longer holds for this problem; consider a line segment $s_i$ and two points (degenerate segments) $s_j$ and $s_k$ on a line which is intersecting $s_i$, so that each of $s_j$ and $s_k$ lies on one side of $s_i$, and $d(s_i,s_j)$ and $d(s_i,s_k)$ is a very small positive constant, as we have illustrated in \Cref{fig:triangle_inequality}.
Observe that $d(s_j,s_k)>d(s_i,s_j)+d(s_i,s_k)$.
 Therefore, the proof of metric $k$-center (including Euclidean $k$-center) no longer applies to this case, and we need to prove the approximation factor.

\begin{figure}[h]
\vspace*{-1mm}
	\centering
	\begin{minipage}{0.40\textwidth}%
		\centering
		\includegraphics[scale=1]{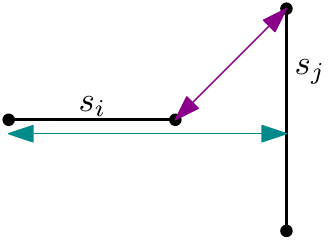}\vspace*{-1mm}
		\caption{The distance $d(s_i,s_j)$ (shown in green) does not equal to $d(s_j,s_i)$ (shown in purple).}
		\label{fig:nonsym}
	\end{minipage}%
	\hfill
	\begin{minipage}{0.50\textwidth}%
		\centering
		\includegraphics[scale=1]{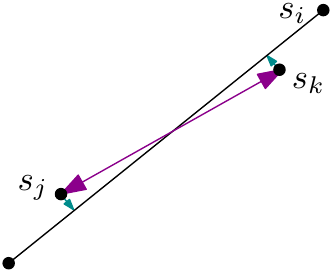}\vspace*{-1mm}
	\caption{The distance $d(s_j,s_k)$ (shown in purple) is larger than $d(s_i,s_j)+d(s_i,s_k)$; and the triangle inequality does not hold.}
\label{fig:triangle_inequality}
	\end{minipage}\vspace*{-6mm}
\end{figure}

\subsection{The maximum $k$-center}

\subsubsection{The hardness of maximum $k$-center}
We define a version of the set cover problem for intervals and prove its hardness. Then, we use it to prove the hardness of $k$-center of segments beyond the lower bound 1.822 from the $k$-center of a set of points~\cite{feder1988optimal}, which holds for polygons because a point is a degenerate polygon.

\begin{figure}[h]
\vspace*{1mm}
	\centering
	\includegraphics[scale=1]{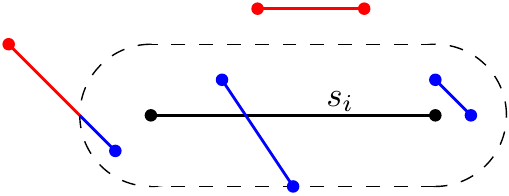}\vspace*{-1mm}
	\caption{The Minkowski sum of a segment $s_i$ with a disk of radius $r$. Only the blue segments are covered by selecting $s_i$ as a center. }
	\label{fig:seg_cluster}\vspace*{-1mm}
\end{figure}

In \Cref{alg:reduction}, $\oplus$ denotes the Minkowski sum.
\begin{algorithm}[h]
	\caption{Reduction from $k$-Center of Segments to Multi-Interval Set Cover}
	\label{alg:reduction}
	\begin{algorithmic}[1]
		\Require{A set of segments $S$, an integer $k$, a constant $r$}
		\Ensure{$k$ segments as centers}
		\For{$s_i \in S$}
		\For{$s_j \in S$}
		\State{$Q_j\gets Q_j \cup \{$ the part of $s_j$ inside $s_i\oplus$ disk of radius $r\}$.}
		\EndFor
		\EndFor
		\State{$C$= the solution of multi-interval set cover with $\{Q_i\}_{i=1}^{n}$ as sets and $S$ as the universal set.}
		\State{return $C$}
	\end{algorithmic}
\end{algorithm}\vspace*{1mm}

\begin{lemma}\label{lemma:reduction}
	The time complexity of \Cref{alg:reduction} is $O(n^2+T(n))$, where $T(n)$ is the time complexity of the multi-interval set cover.
\end{lemma}
\begin{proof}
	Computing the Minkowski sum of a segment with a disk of radius $r$ takes $O(1)$ time (see \Cref{fig:seg_cluster} for an illustration), and finding the intersections between the resulting shape with $n-1$ segments takes $O(n)$ time.
	 The nested for loops are repeated $O(n^2)$ times, each with $O(1)$ time. The last step of the algorithm runs an instance of the multi-interval set cover.
	So, the overall time complexity of the algorithm is $O(n^2+T(n))$.
\end{proof}

\begin{theorem}\label{theorem:aprx1}
	\Cref{alg:reduction} computes a solution with $\alpha k$ centers, if $C$ is an $\alpha$-approximation of multi-interval set cover.
\end{theorem}
\begin{proof}
The set $Q_i$ is the set of segments covered by radius $r$ of a segment $s_i$. So, the minimum number of the sets $Q_i, i=1,\ldots,n$ that covers all other segments is the set of centers $s_i$ with the same indices. So, the number of centers is the number of sets in $C$. So, an $\alpha$-approximation for multi-interval set cover, is a solution with $\alpha k$ centers and radius $r$.
\end{proof}

Based on \Cref{alg:reduction}, in the clustering problem for segments, the ply of the intervals is equal to the number of segments that can cover the same point on a segment.

\subsubsection{The hardness of multi-interval set cover}
In this case, each member of the universal set is an interval and each set of covers is also a set of intervals.

\begin{figure}[h]
	\centering
	\includegraphics[scale=1.1]{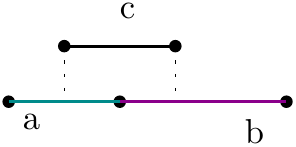}\vspace*{-2mm}
	\caption{$1$-Dimensional intervals $a,b,c$, such that $c\subset (a\cup b)$.}
	\label{fig:intervals}
\end{figure}

For three intervals $a,b,c$ such that $c$ is a subset of the union of $a$ and $b$ (See \Cref{fig:intervals}), we have:
\[
\{a\} \cup \{b\} = \{a\} \cup \{b\} \cup \{c\}.
\]
Therefore, the set cover of sets of intervals is different from the set cover of sets of points.

\begin{theorem}\label{theorem:npc}
The	multi-interval set cover is NP-complete.
\end{theorem}
\begin{proof}
	The decision version of multi-interval set cover can be solved by sorting the intervals in the sets and sweeping the intervals in the universe while keeping the last covered point of the interval. So, the problem is in NP. It remains to prove the NP-hardness.
	
\medskip
	We proceed by a reduction from set cover.
	Let $U$ be the universal set and $\mathcal{S}_1,\ldots,\mathcal{S}_n$ be the set of sets in the set cover instance.
	Map the points of $U$ to the interval $[0,n]$ on the real line, where the interval $[i-1,i]$ corresponds to the $i$-the member of $U$.
	Any set $\mathcal{S}_i$ can then be represented as a set of intervals; see \Cref{fig:seg_cluster} for an illustration.
	This is an instance of the multi-interval set cover problem with $[0,n]$ as the universe and the intervals of $\mathcal{S}_i, i=1,\ldots,n$ as the sets of intervals.
	Conversely, it is easy to see that the optimal solution of this instance of multi-interval set cover results in an optimal solution of the set cover. So, the current multi-interval set cover instance is equivalent to set cover.
\end{proof}

\begin{theorem}\label{theorem:preserving}
	The reduction of \Cref{theorem:npc} is approximation-preserving.
\end{theorem}
\begin{proof}
	In the multi-interval set cover, the objective function is the number of segments in the universal set that are covered by the chosen sets. Each member of the universe in the set cover is an interval in the multi-interval set cover in the reduction of \Cref{theorem:npc}.
	Therefore, the objective function of set cover is preserved by the reduction.
\end{proof}

\subsubsection*{Approximating multi-interval set cover}\label{sec:intervalcover}

Based on \Cref{theorem:npc}, the best possible approximation factor for this problem is $\Omega(\log n)$, as well~\cite{vazirani2013approximation}. For low frequency set cover, $f$-approximation, where $f$ is the maximum frequency of an element is possible~\cite{vazirani2013approximation}.\medskip

\begin{algorithm}[h]
	\caption{Approximate Multi-Interval Set Cover}
	\label{alg:intervalcover}
	\begin{algorithmic}[1]
		\Require{A set of 1D intervals $Q_i,i=1,\ldots,n$}
		\Ensure{A subset of $Q$}
		\State{$U$= the set of disjoint intervals with endpoints from the endpoints of intervals in $Q_i, i=1,\ldots,n$}
		\State{$S_i$= the set of intervals in $U$ that are covered by the intervals of $Q_i,i=1,\ldots,n$}
		\State{Solve set cover for $S_i,i=1,\cdots,n$}
		\\ \Return{the indices of the sets from the previous step.}
	\end{algorithmic}
\end{algorithm}

\begin{theorem}\label{theorem:intervalsetcover}
	\Cref{alg:intervalcover} solves the multi-interval set cover problem with the same approximation factor as the set cover.
\end{theorem}

\begin{proof}
	Based on the construction of set $U$, the covered intervals are the same as that of $\cup_{i=1}^n Q_i$, i.e. $\cup_{u\in U} u=\cup_{i=1}^n \cup_{q\in Q_i} q$.
	The construction of $S_i$ also shows that $\cup_{s\in S_i} s=\cup_{q\in Q_i} q$.
	So, a set cover for $\cup_{u\in U} u$ and $S_i, i=1,\ldots,n$ is a set cover for $\cup_{i=1}^n \cup_{q\in Q_i} q$ and $Q_i$, for $i=1,\cdots,n$.
	
\medskip
	The opposite also holds, since the unions of intervals in each of the discussed sets are equal to their corresponding sets in the other instance.
	As a result, the objective function, the solution set and the approximation factor of these problems are preserved.
\end{proof}
\begin{lemma}\label{lemma:cover}
	\Cref{alg:reduction} takes $O(n^3)$ time, if the $O(\log n)$-approximation algorithm \cite{vazirani2013approximation} is used.
\end{lemma}
\begin{proof}
	The universe is the union of the intervals resulting from cutting each of the $n$ input intervals at the endpoints of other $n-1$ intervals. So, it has size at most $O(n^2)$.
	The number of steps of the set cover algorithm is at most $n$, since at each step of the algorithm at least one set is chosen.
	Removing the elements of the universe covered at each step takes $O(n^2)$ time.
	So, the overall time complexity of the algorithm is $O(n^3)$.
\end{proof}
The time complexity of the algorithm for low frequency (equivalent to low ply in this setting) is the complexity of solving the corresponding linear program~\cite{vazirani2013approximation}, therefore, it is different from the time complexity of the $O(\log n)$-approximation algorithm discussed in \Cref{lemma:cover}.

\subsubsection{Approximating the radius of $k$-center}

\paragraph{1-center} In the case where the input is a set of segments and $k=1$, the problem is equivalent to compute the SED of the segments whose center is restricted to lie on a segment. The SED in this case be determined by the endpoints of the segments.

\begin{algorithm}[h]
	\caption{$1$-Center}
	\label{alg:1center}
	\begin{algorithmic}[1]
		\Require{A set of segments $S$}
		\Ensure{A center from $P$}
		\For{$p\in P$}
		\State{$d_p=\max_{s\in P} \max_{q\in s} d(p,q)$}
		\EndFor
		\\ \Return{$\arg \min_{p\in P} d_p$.}
	\end{algorithmic}
\end{algorithm}
Computing the distance from a segment to $n-1$ other segment takes $O(n)$ time.
\Cref{alg:1center} for each input segment, computes its distance to all other segments, so, the total running time of this algorithm is $O(n^2)$.

\medskip
Now, we compute a set of candidates for approximations of the radius of the $k$-center of a set of segments. These values are used in the latter sections in a parametric pruning algorithm to find the corresponding set of centers.
Note that $\epsilon$ needs to be smaller than $d(s_i,s_j)/2$, for any $s_i,s_j \in S$.
Computing the smallest pairwise distance takes $O(n \log n)$ time by constructing a Voronoi diagram for segments (Line 1).

\begin{algorithm}[h]
	\caption{Finding $r$}
	\label{alg:findr}
	\begin{algorithmic}[1]
		\Require{A set of segments $S$, an integer $k$, a constant $\epsilon$}
		\Ensure{A set of radius $R$}
		\State{Discretize the segments of $S$ with chunks of radius $\epsilon_0$}
		\State{$R$= the pairwise distances between any pair of vertices from the previous step.}
		\\ \Return{$R$}
	\end{algorithmic}
\end{algorithm}

\begin{theorem}\label{theorem:radius}
	\Cref{alg:findr} finds a set that solving the $k$-center problem on that gives a $(1+\epsilon)$-approximation for the radius of the $k$-center of a set of segments.
\end{theorem}
\begin{proof}
	The distance between an optimal center and a member of its cluster is modified by at most $2\epsilon$ by discretization, assuming that $\epsilon$ is smaller than the minimum of $R$. The optimal cost is the distance between an endpoint of a segment with a point on another segment. Therefore, after discretization, this distance is at most multiplied by $1+2\epsilon$. Scaling $\epsilon$ by a constant factor concludes the proof.
\end{proof}

\begin{lemma}\label{lemma:findr}
	\Cref{alg:findr} takes $O(\frac{1}{\epsilon^4})$ time.
\end{lemma}
\begin{proof}
	The number of points after discretization is $O(\frac{1}{\epsilon^2})$. Taking pairs of these distances can be done in $\binom{O(\frac{1}{\epsilon^2})}{2}=O(\frac{1}{\epsilon^4})$ ways. Computing the distance between a pair takes $O(1)$ time.
	So, the time complexity of the algorithm is $O(\frac{1}{\epsilon^4})$.
\end{proof}

\subsubsection*{Approximate $k$-center}
We discretize the segments to find the candidates for $r$ to use the set cover algorithm.

\begin{algorithm}[h]
	\caption{Approximate $k$-Center of Segments}
	\label{alg:new}
	\begin{algorithmic}[1]
		\Require{A set of segments $S$, an integer $k$, a constant $\epsilon$}
		\Ensure{A set of $k$ segments as centers $C$}
		
		\State{$R=$ Run \Cref{alg:findr}.}
		\For{$r\in R$}
		\State{$C=$ Run \Cref{alg:reduction} using the set cover algorithm of \cite{vazirani2013approximation}.}
		\EndFor
		\\ \Return{$C$}
	\end{algorithmic}
\end{algorithm}

\begin{theorem}\label{theorem:new}
	\Cref{alg:new} gives a solution with $O(k\log n)$ centers and radius $r(1+\epsilon)$ for $k$-center of segments.
\end{theorem}
\begin{proof}
	Using \Cref{theorem:intervalsetcover}, the $O(\log n)$-approximation for set cover can be used to solve the multi-interval set cover with the same approximation factor.
	This means that \Cref{alg:new} finds at most $O(k\log n)$ centers, according to\Cref{theorem:aprx1}.
	The set of radii from \Cref{alg:findr} contain a $(1+\epsilon)$-approximation for the optimal radius as proved in \Cref{theorem:radius}.
\end{proof}

\begin{theorem}\label{theorem:time}
	\Cref{alg:new} takes $O(\frac{n^3}{\epsilon^4})$ time.
\end{theorem}
\begin{proof}
	The time complexity of the algorithm is the time complexity of \Cref{alg:findr} which is $O(\frac{1}{\epsilon^4})$ based on \Cref{lemma:findr}, plus $|R|=O(\frac{1}{\epsilon^4})$ times the complexity of \Cref{alg:reduction} using the geometric set cover, which is $O(n^3)$. So, the algorithm takes $O(\frac{n^3}{\epsilon^4})$ time.
\end{proof}

\subsection{Minimum cost $k$-center }

In the case where $k=1$, the objective is to find a segment $s_i$ such that the maximum distance of any point on any other segment to $s_i$ is minimized. This problem is equivalent to compute a radius $r$ for which the Minkowski sum of a segment with the disk of radius $r$ intersects the other segments.
		
to do: add the Minkowski sum step to algorithm 6

\begin{algorithm}[h]
	\caption{Minimum $1$-Center of Segments}
	\label{alg:m1center}
	\begin{algorithmic}[1]
		\Require{A set of segments $P$}
		\Ensure{A center from $P$}
		\For{$p\in P$}
		\State{$d_p=\max_{s\in P} \min_{s\in q} d(p,q)$}
		\EndFor
		\\ \Return{$\arg \min_{p\in P} d_p$.}
	\end{algorithmic}
\end{algorithm}

\begin{theorem}
	\Cref{alg:m1center} finds a minimum $1$-center in $O(n^2)$ time.
\end{theorem}

\begin{proof}
	The time required for computing the radius $r$ for which the Minkowski sum of a segment $p$ with the disk of radius $r$ intersects all segments is $O(n)$. So, the overall time complexity of the algorithm is $O(n^2)$.
\end{proof}

We modify the algorithms for $k$-center of segments (\Cref{alg:new}) to solve the minimum $k$-center of segments problem.

\medskip
As before, in \Cref{alg:reduction2}, $\oplus$ denotes the Minkowski sum.
\begin{algorithm}[h]
	\caption{Reduction from Minimum $k$-Center of Segments to Multi-Interval Set Cover}
	\label{alg:reduction2}
	\begin{algorithmic}[1]
		\Require{A set of segments $S$, an integer $k$, a constant $r$}
		\Ensure{$k$ segments as centers}
		\For{$s_i \in S$}
		\For{$s_j \in S$}
		\If{$s_i\oplus $ disk of radius $r$ intersects with $s_j$}
		\State{$Q_j\gets Q_j \cup \{s_j\}$.}
		\EndIf
		\EndFor
		\EndFor
		\State{$C$= multi-interval set cover with $\{Q_i\}_{i=1}^{n}$ as sets and $S$ as the universal set.}
		\State{return $C$}
	\end{algorithmic}
\end{algorithm}

\begin{algorithm}[h]
	\caption{Approximate Minimum $k$-Center of Segments}
	\label{alg:new2}
	\begin{algorithmic}[1]
		\Require{A set of segments $S$, an integer $k$, a constant $\epsilon$}
		\Ensure{A set of $k$ segments as centers $C$}
		\State{$R=$ Run \Cref{alg:findr}.}
		\For{$r\in R$}
		\State{$C=$ Run \Cref{alg:reduction2} using the set cover algorithm of \cite{vazirani2013approximation}.}
		\EndFor
		\\ \Return{$C$}
	\end{algorithmic}
\end{algorithm}

\begin{theorem}
	\Cref{alg:new2} gives a solution with $O(k\log n)$ centers and radius $r(1+\epsilon)$ for $k$-center of segments.
\end{theorem}

\begin{proof}
	The proof is similar to \Cref{theorem:new}, the only difference is that it is enough if one point of a segment is within distance $r$ of a segment, which was modeled in \Cref{alg:reduction2}.
\end{proof}

\smallskip\noindent
\textbf{Remark.} We note that since a segment can be considered as a polygon with zero area, the presented algorithms in the next section can be applied to solve the MinMax segment clustering problem in the AGU model. The difference is that the centers are segments in the MinMax segment clustering, while they are points in the domain-restricted $k$-center of polygons.

\section{Domain-restricted $k$-center of polygons}

\subsection{Maximum $k$-center}
For the maximum $k$-center of polygons for $k=1$, the center of SED of the vertices might not fall inside any of the polygons. We give a $2$-approximation algorithm for this problem (\Cref{alg:1cluster}).\vspace*{1mm}
\begin{algorithm}[h]
	\caption{$1$-center}
	\label{alg:1cluster}
	\begin{algorithmic}[1]
		\Require{A set of polygons $P$}
		\Ensure{A center from $P$}
		\State{$c=$ the center of the SED of points of $P$ and the vertices of the polygons in $P$}
		\\ \Return{the nearest neighbor of $c$ in $P$}
	\end{algorithmic}\vspace*{-2mm}
\end{algorithm}

\begin{theorem} \label{thm:clustercenter}
	The approximation factor of \Cref{alg:1cluster} is $2$.
\end{theorem}
\begin{proof}
	By definition, the radius of the smallest enclosing disk $(r)$ is a lower bound for $1$-cluster. If the center of the smallest enclosing disk falls inside the polygon, that is the optimal solution.
	The distance from any point to the center of the smallest enclosing disk is at most its radius. Let $o$ denote the center of the optimal $1$-center and $p$ be the nearest neighbor (point) of $c$ in $P$. Then, using triangle inequality,
	\[
	d(p,o) \leq d(o,c)+d(c,p) \leq 2r.
	\]

\vspace*{-7mm}
\end{proof}

\begin{theorem}
	The time complexity of \Cref{alg:1cluster} is $O(n)$ expected time and $O(n\log n)$ worst-case time.
\end{theorem}
\begin{proof}
	Computing the nearest neighbor of a point to a set of segments takes $O(n)$ time. Also, the SED can be computed in expected linear time and worst-case $O(n\log n)$ time, so the total time complexity is $O(n)$ expected time and $O(n\log n)$ worst-case time.
\end{proof}

\subsubsection{Maximum $k$-center of convex polygons}

	If we solve the problem with a naive grid algorithm, we may miss the covering of the sharp corners of the boundary of the polygons. To treat these cases, we first compute the Minkowski sum of a polygon with a square of side length $\epsilon$, and then consider the input on a grid with cell length $\epsilon$. This guarantees at least one vertex is chosen for the sharp corners. See \Cref{fig:triangle} for an illustration.
	But the rounded polygons on the grid may impose some centers which do not lie on the polygons. To avoid this, we map the vertices of the rounded polygons to the closest point of the polygon to satisfy the domain-restriction assumption (Line 11 and 12 of \Cref{alg:kcluster}). We note that $\epsilon$ needs to be smaller than $\frac{r_i}{k}$, where $r_i$ denotes the radius of the SEC of the vertices of $P_i$, for $i=1,\ldots,n$.
	
\begin{figure}[h]
\vspace*{2mm}
	\centering
	\includegraphics[scale=1]{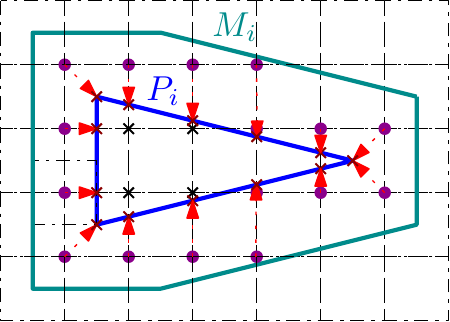}
	\caption{A polygon $P_i$ and the result of the Minkowski sum of $P_i$ and a square of side length $\epsilon$. The cross points denote the contributed points by $P_i$ on the discretization.}
	\label{fig:triangle}\vspace*{-2mm}
\end{figure}

\medskip
	We show that applying any $k$-center algorithm on the union of the grid points lying within the polygons and the mapped points of the grid to the boundary of the polygons (in sharp corners) solves our problem. If we use the presented algorithm of~\cite{vazirani2013approximation}, we come up with a $2+\epsilon$ approximation.

\begin{algorithm}[h]
	\caption{$k$-Center of Convex Polygons}
	\label{alg:kcluster}
	\begin{algorithmic}[1]
		\Require{A set of polygons $P$, an integer $k$, a constant $\epsilon$}
		\Ensure{A set of $k$ centers inside $P$}
		\For{$i=1,\ldots,n$}
		\State{$M_i=$ the Minkowski sum of $P_i$ with a disk of radius $\epsilon$.}
		\EndFor
		\State{build a grid of cell length $\epsilon$}
		\State{$X=$ compute the vertices of the grid inside the shapes in $M$}
		\State{$S=\emptyset$}
		\For{$x\in X$}
		\If{$x \notin P$}
		\For{$M_i\in M, x\in M_i$}
		\State{$Y=\emptyset$}
		\State{$y=$ the nearest neighbor of $x$ in $P_i$.}
		\State{$Y\gets Y \cup \{y\}$ }
		\EndFor
		\EndIf
		\State{$S\gets S\cup Y$.}
		\EndFor
		\State{$C=$ an approximate $k$-center of $S \cup X$}
		\State{return $C $}
	\end{algorithmic}
\end{algorithm}

\begin{theorem}\label{lemma:triangle}
	The approximation factor of \Cref{alg:kcluster} is $2+\epsilon$, for fixed aspect ratio and $\epsilon$ and its time complexity is $O(\sum_{i=1}^n |P_i|+\frac{1}{\epsilon^2})$.
\end{theorem}

\begin{proof}
	Since $\epsilon \leq r_i$, at least one vertex of a grid with cell length $\epsilon$ falls inside $M_i$. Therefore, all the points of $P_i$ are covered by disks of radius $\epsilon$ centered at the vertices of the grid.
	Mapping a grid vertex $(x)$ to its nearest neighbor in $P_i$ $(y)$ still covers the same area of $P_i$, but with radius $2\epsilon$.
	So, a clustering of $Y$ with radius $r$, is a clustering of $X$ with radius $r+\epsilon$, and a clustering of $P_i$ with radius $r+2\epsilon$.
	Let $r^*$ denote the optimal radius of the $k$-center of $\cup_{i=1}^n P_i$.
	Since $Y$ is a subset of the points of polygons, the radius of $k$-center of $Y$ is less than or equal to the $k$-center radius of $\cup_{i=1}^n P_i$, i.e. $r\leq r^*$.
	The approximation factor of the algorithm is therefore:
	\[
	\frac{r+2\epsilon}{r^*} \leq 1+\frac{2\epsilon}{r^*} \leq 1+2\epsilon.
	\]
	
		For the smallest possible $r_i$, $i=1,\ldots,n$ we have
	 $\frac{r_i}{k}\leq 2r^*$
	 which is because each $P_i$ can be covered by multiple centers, and the radius of the clusters are at least $\frac{r_i}{2k}$. Knowing
	 $\epsilon \leq \min_i \frac{r_i}{k}$, we prove the right side of the inequality.
	Since we compute a $2$-approximation of the $k$-center of $Y$, the overall approximation factor is $2(1+2\epsilon)=2+4\epsilon$.

\medskip	
	Computing the circumscribed circle of $P_i$ takes $O(|P_i|)$ time. So, computing the minimum of the radii of these circles takes $O(\sum_{i=1}^n |P_i|)$.
	The time complexity of computing a $k$-center for the set of grid points is $O(k|Y|)=O(k|X|)$. Also, we have that:
$
	|X|=\sum_{i=1}^n \frac{\mathrm{Area}(P_i)}{\epsilon^2},
$
	where $\mathrm{Area}(P_i)$ denotes the area of polygon $P_i$.
	So, the total time complexity of the algorithm is:
	$
	O(\sum_{i=1}^n |P_i|+ \frac{\sum_{i=1}^n \mathrm{Area}(P_i)}{\epsilon^2}).
$
	By scaling the bounding box of the input to have a unit area, i.e. $\sum_{i=1}^n \mathrm{Area}(P_i)=1$, and adjusting the value of $\epsilon$ to be $\epsilon'/4$, where $\epsilon'$ is the input value for $\epsilon$, the bounds in the statement of the theorem are achieved.
	\end{proof}

\subsubsection{Maximum $k$-center of abitrary polygons}

 For a set of arbitrary polygons,
	we first compute a triangulation of $P_i$ for $i=1,\ldots,n$. Let $T$ denote the set of all triangles of $P_i$ for $i=1,\ldots,n$.
	We apply 	\Cref{alg:kcluster} on the set $T$, and compute a discrete set of $\Theta(nk)$ points. We show that applying a $k$-center algorithm with (e.g., with approximation factor 2) on the discrete set gives a $6+\epsilon$ approximation of the optimal solution.
	
	\begin{algorithm}[h]
		\caption{Approximation Algorithm for the $k$-Center of arbitrary polygons}
		\label{alg:kclusterarbitrary}
		\begin{algorithmic}[1]
			\Require{A set of polygons $P$, an integer $k$}
			\Ensure{A set of centers from $P$}
			\State{$T$= the set of triangles in a triangulation of the input polygons}
		\State{$X$= the output of \Cref{alg:kcluster} on the set $T$}
		\State{$C=$ approximate $k$-center of $X$}
		\State{{$C'=$ the closest point of each center in $C$ to a point of a polygon of $P$}}
			\State{return $C$}
		\end{algorithmic}
	\end{algorithm}

\begin{theorem} \label{theorem:arbitrary_polygons}
	The approximation factor of \Cref{alg:kclusterarbitrary} is $6+\epsilon$.
\end{theorem}

\begin{proof}
Let $\mathcal{R}(Z)$ denote the optimal radius of the $k$-center of
the points inside a set of shapes $Z$. If $Z$ is a set of points, $\mathcal{R}(Z)$ is the optimal radius of the points in $Z$. If $Z$ is a polygon, $\mathcal{R}(Z)$ denote the radius of the $k$-center of the points inside $Z$. If $Z$ is a set of polygons, $\mathcal{R}(Z)$ denote the optimal radius of the $k$-center of all the points inside the polygons $Z$.
The optimal centers of the $k$-center of a convex polygon lie inside it, since otherwise there is a center inside it with a lower or equal radius.
The set of points of each triangle $T_i\in T$ is a subset of the points of $P$, so, they are covered by the centers of $P$ with radius at most $\mathcal{R}(P)$. Based on the optimality of $\mathcal{R}(T_i)$, we have:
\[ \mathcal{R}(T_i)\le \mathcal{R}(P). \]
Taking the maximum over all such triangles gives the following bound:
\[ \max_{\forall T_i \in T} \mathcal{R}(T_i) \le \mathcal{R}(P).\]

Similarly, the centers of the $k$-center of $T_i$'s lie inside $P$, so, $X \subseteq P$ and we have $\mathcal{R}(X) \le \mathcal{R}(P)$.

\medskip
Running \Cref{thm:clustercenter} with approximation factor $2+\epsilon$ on the triangles, gives the set $X$.
 To guarantee the centers chosen by the algorithm lie inside the input polygonal domain, we compute a $k$-center on set $X$, where only points of $X$ can be chosen as centers. The approximation factor of this problem with the $k$-center using arbitrary points of the plane as centers is $2$, and $\mathcal{R}(X)\le \mathcal{R}(T)$.
Using a $2$-approximation $k$-center algorithm on $X$, results in a 4-approximation, and since the approximation factor of \Cref{alg:kcluster} is $2+\epsilon$, \Cref{alg:kclusterarbitrary}, and applying the triangle inequality gives a $(6+\epsilon)$-approximation for $\mathcal{R}(P)$:
\[
\max_{T_i\in T} \mathcal{R}(T_i)+2\mathcal{R}(X)\leq
(2+\epsilon)\mathcal{R}(P)+2( 2\mathcal{R}(P))=(6+\epsilon)\mathcal{R}(P).
\]

\vspace*{-7mm}
\end{proof}	

\subsection{Minimum cost $k$-center}

In \Cref{alg:minkcenterC}, first we build a grid similar to the maximum $k$-center (\Cref{alg:kcluster}) with resolution $\epsilon= \min (\epsilon, \min_{p,q\in P} \min_{u\in p, v\in q} d(u,v))$, then we use colorful $k$-center to find a $k$-center that covers at least one point from each polygon. The limitation is that colorful $k$-center has only been solved for a constant number of colors in polynomial time~\cite{bandyapadhyay2019constant}, which limits the number of input polygons to a constant.
\begin{algorithm}[h]
	\caption{Approximate Minimum $k$-Center of Polygons}
	\label{alg:minkcenterC}
	\begin{algorithmic}[1]
		\Require{A constant size set of polygons $P$, an integer $k$, a constant $\epsilon>0$}
		\Ensure{A set of centers}
		\State{Build a grid of cell length $\epsilon$}
		\State{$M_i = P_i\oplus$ disk of radius $\epsilon$}
		\State{$T_i=$ the closest point of $\cup_j P_j$ to the vertices of the grid inside $M_i$, for $P_i\in P$}
		\State{Color the points in $T_i$, for $P_i\in P$ with color $i$.}
		\State{$c=$ compute the colorful $k$-center of the colored points.}
		\\ \Return{$c$}
	\end{algorithmic}
\end{algorithm}

\begin{theorem}
	\Cref{alg:minkcenterC} is a $O(1)$-approximation for minimum $k$-center.
\end{theorem}
\begin{proof}
	To guarantee that the vertices of the grid cover the area of each polygon with distance at most $O(\epsilon)$, we compute the grid vertices inside the Minkowski sum of each polygon with the disk of radius $\epsilon$.
	To satisfy the constraint that the centers must be a subset of input polygons, we replace these points with their nearest neighbors in the input polygons.
	Computing the solution on the vertices of a grid adds a $1+\epsilon$ factor to the approximation.
	The approximation factor of colorful $k$-center is $O(1)$ for a constant number of colors.
\end{proof}

\section{Experimental studies}
We implement our algorithm for maximum $k$-center in the aggregated uncertainty model (\Cref{alg:kclusterarbitrary}) and compare the results to the $k$-center of points (\cite{malkomes2015fast}) on a big network data-set~\cite{cho2011friendship,snapnets}, for $k=20$ and $\epsilon=5$, where $\epsilon$ is the grid cell length.

\medskip
The approximation factors of the $k$-center algorithms for big data is 4 for the metric case \cite{malkomes2015fast}, and $2+\epsilon$ for the doubling metrics, including low-dimensional Euclidean space \cite{ceccarello2019solving,aghamolaei2018composable}.
Since we want a small summary size, we use the $4$-approximation algorithm with a summary of size $O(kL)$ instead of the $(2+\epsilon)$-approximation algorithms with summaries of size $\omega(\frac{kL}{\epsilon^2})$, where $L$ is the number of partitions. \Cref{alg:kclusterarbitrary} creates a summary of size $O(\frac{L}{\epsilon^2})$, which is still feasible in practice. Note that $\epsilon$ is about $1/10$ times the radius in our experiment and halving the value of $\epsilon$ can multiply the size of the summary by at most the doubling dimension of the Euclidean plane, which is $5$.

\medskip
The data-set used in our experiments is the time and location information of check-ins made by users of the Brightkite social network. It has 58,228 users and 4,491,143 check-ins of these users over the period of Apr. 2008 - Oct. 2010.
It is available as a part of the SNAP network data-sets~\cite{snapnets}.
The number of users with at least one check-in the aforementioned period of time is 51,685. Each user had at most 325,821 check-ins.

\medskip
In then implementation of \Cref{alg:kclusterarbitrary}, we also
use a test data set, that contains the grid points which lie inside the convex polygons. We compare the output of our algorithm on both the test data set and the convex polygons itself.

We consider each check-in data as a point in 2D space.
We use the maximum $k$-center algorithm of polygons to cluster the users into $k$ groups.
To summarize the data of each user, we compute their convex hull. As the test data-set, we use the vertices of the grid of cell-length $\epsilon$ that fall inside each convex hull as possible future locations of that user.
After computing the convex-hulls, the total number of vertices was 189,355, which is an almost 23.7\% data-compression ratio.

 \begin{figure}[!h]
 \vspace*{1mm}
 	\centering
 	\includegraphics[scale=0.49]{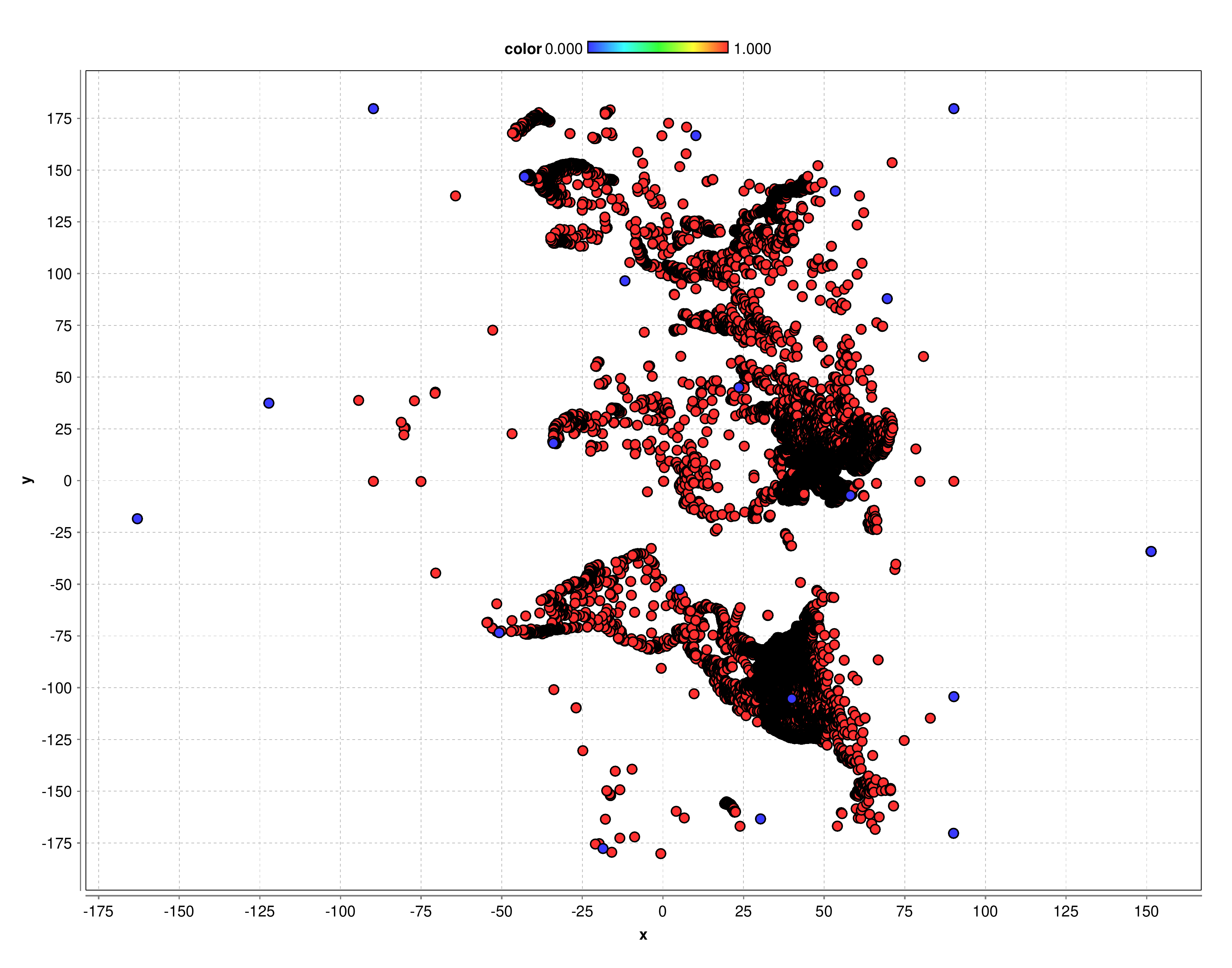}\vspace*{-5mm}
 	\caption{The approximate $k$-center of vertices using the $4$-approximation $k$-center algorithm for points~\cite{malkomes2015fast} for $k=20$ (blue). The red points are the summary of points for $k=20$ after the first round, which are the union of $k$-centers of the partitions.}
 	\label{fig:composable_gonz}
 \end{figure}

  \begin{figure}[!h]
 	\centering
 	\includegraphics[scale=0.49]{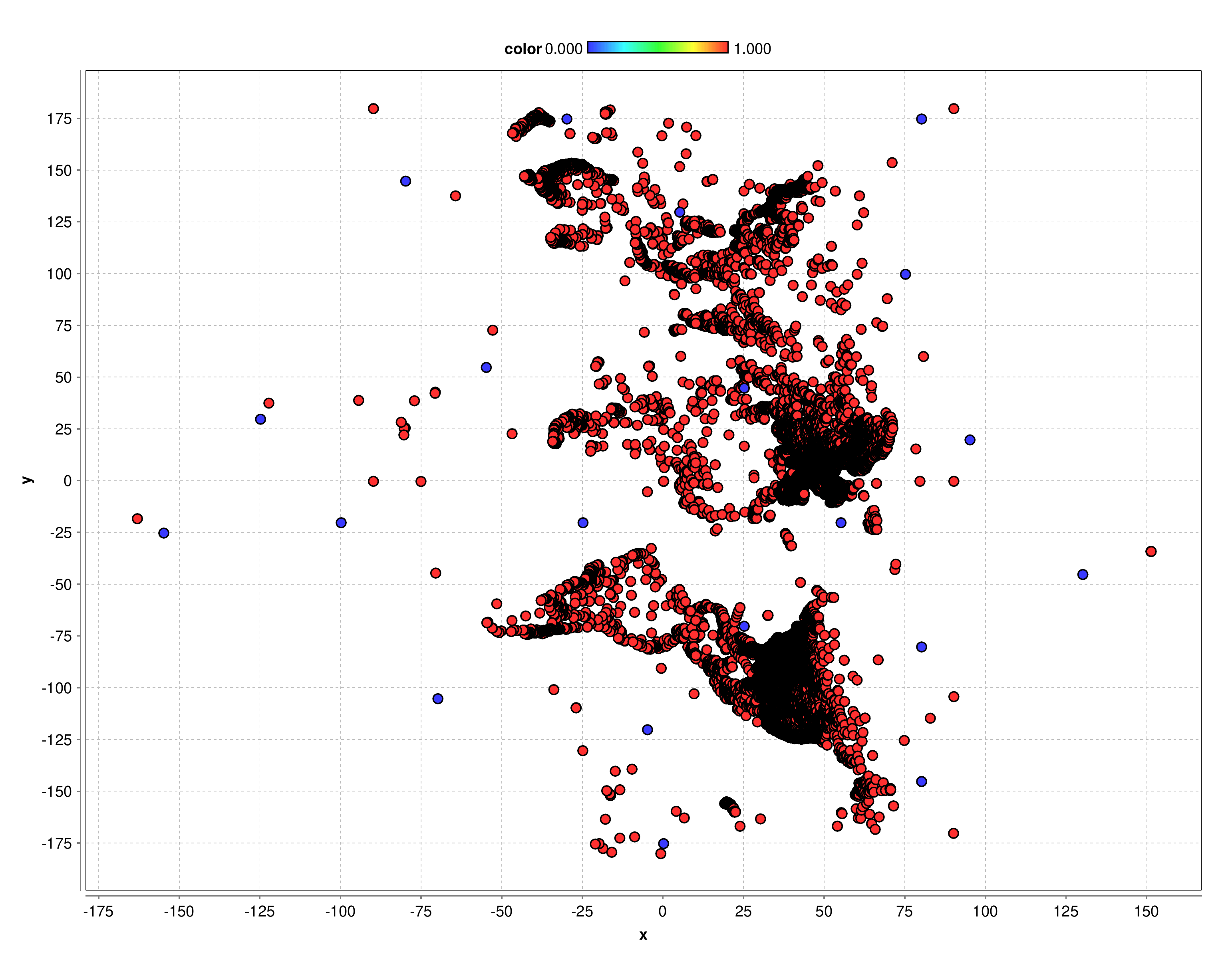}\vspace*{-5mm}
 	\caption{The approximate $k$-center of polygons computed by sampling equidistant points (grid vertices) inside the shapes and then clustering them~\Cref{alg:kclusterarbitrary}, for $k=20$ and $\epsilon=5$ (blue). The red points are the summary of $4$-approximation $k$-center of points~\cite{malkomes2015fast} for $k=20$ and show the centers computed by this algorithm are good centers for the summary points computed in the first round of~\cite{malkomes2015fast}.}
 	\label{fig:composable_grid}
 \end{figure}

In \Cref{fig:composable_gonz}, the summary and the centers of the algorithm of \cite{malkomes2015fast} are shown in red and blue, respectively.
 \Cref{fig:composable_grid} shows the centers of \Cref{alg:kclusterarbitrary} and the summary of the algorithm of \cite{malkomes2015fast}.
In \Cref{fig:grid_image}, the output of \Cref{alg:kclusterarbitrary} for clustering the convex polygons is depicted. The red points are the vertices of the grid inside the convex hull of the convex polygons, and the blue points are the centers of the algorithm.
The result of the clustering of the points of the grid (test data set) with the centers computed by the algorithm presented in~\cite{malkomes2015fast} is illustrated in \Cref{fig:grid_gonz}.

\begin{figure}[!h]
	\centering
		\includegraphics[scale=0.49]{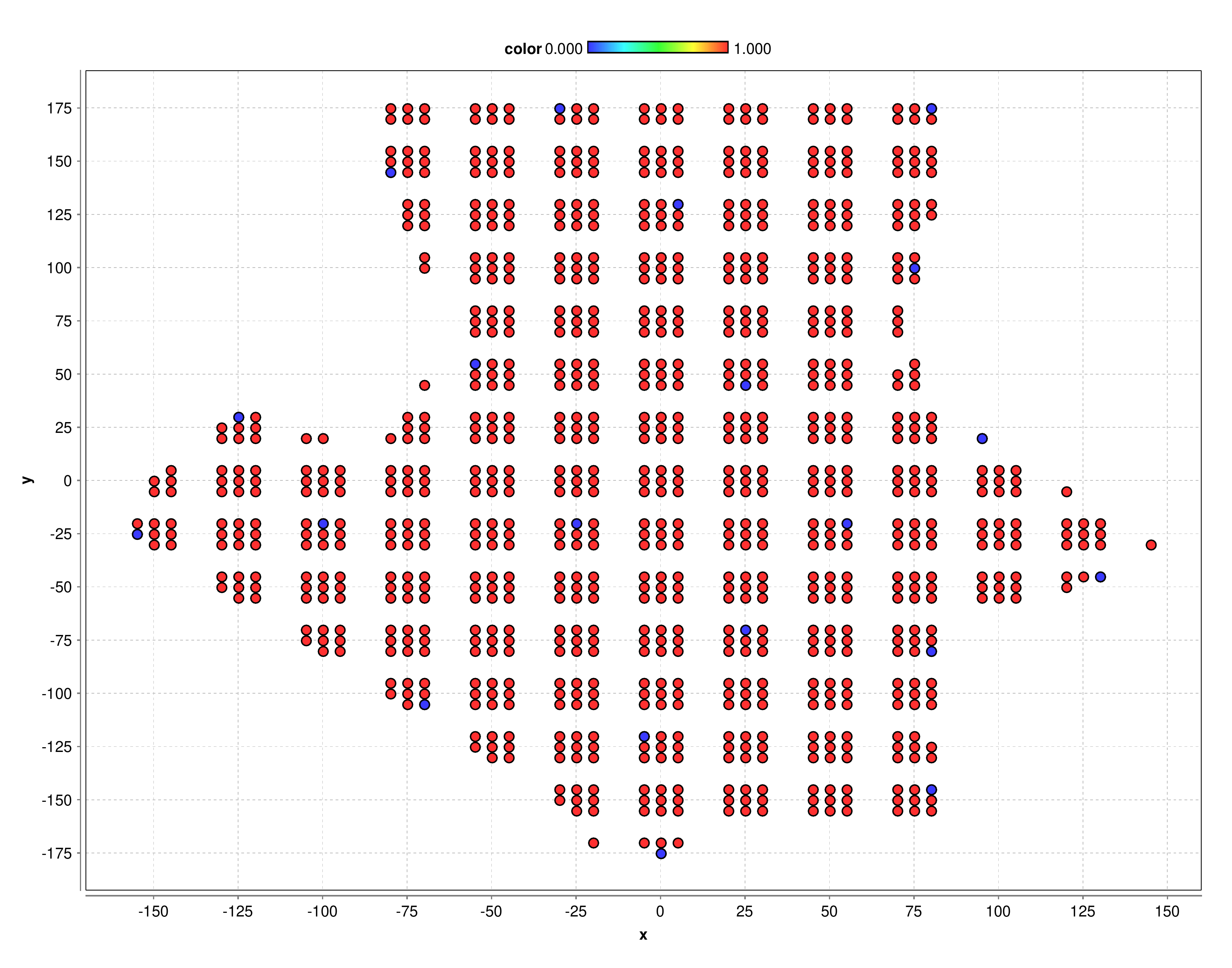}\vspace*{-4mm}
	\caption{The approximate $k$-center of polygons by sampling equidistant points (grid vertices) inside the shapes and then clustering them \Cref{alg:kclusterarbitrary}, for $k=20$ and $\epsilon=5$ (blue). The red points are the summary of~\Cref{alg:kclusterarbitrary} for $\epsilon=5$ that are shown in the figure as the representatives of the input shapes which show the polygonal shape of the input.}
	\label{fig:grid_image}
\end{figure}

\begin{figure}[h]
\vspace*{-1mm}
	\centering
	\includegraphics[scale=0.47]{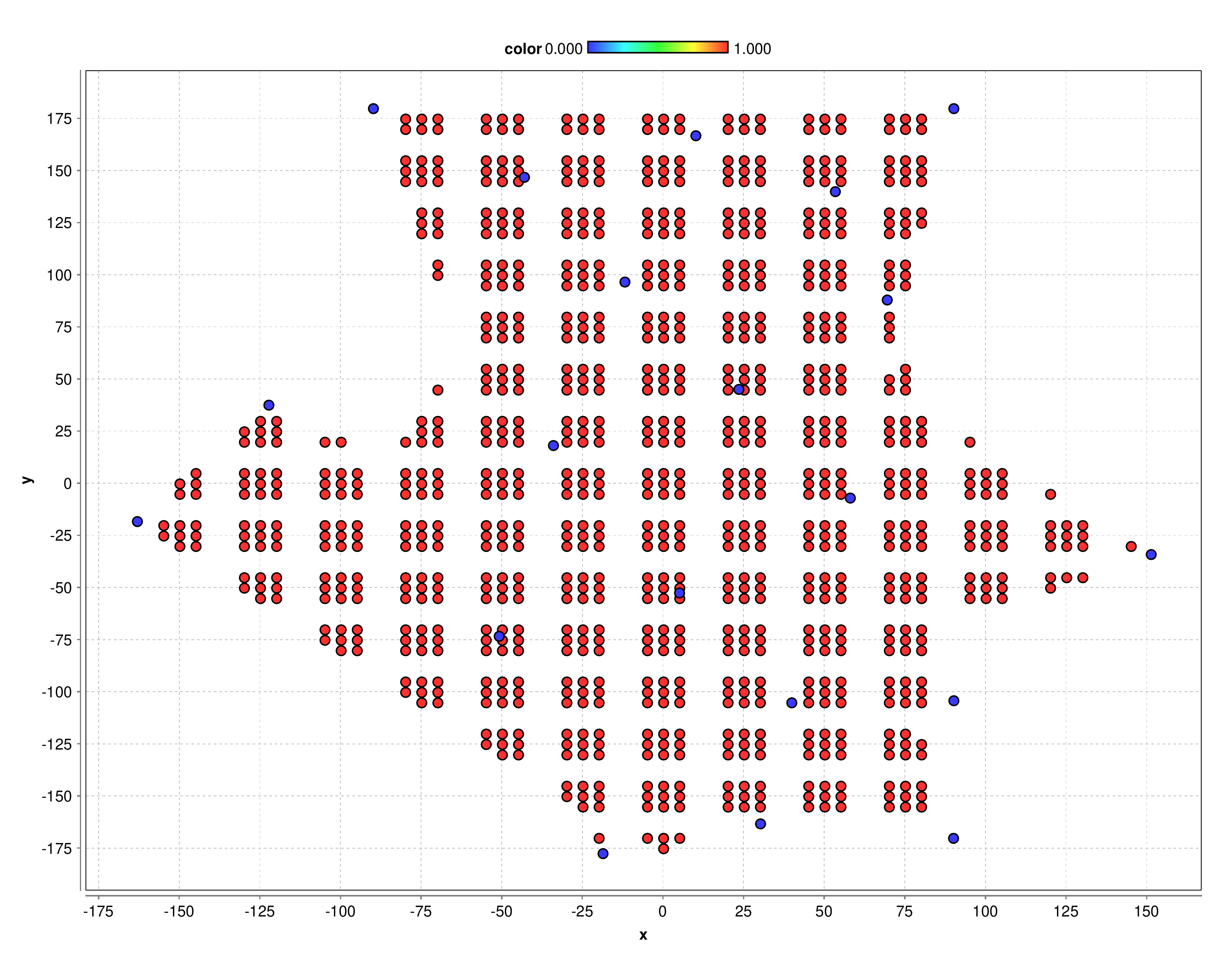}\vspace*{-5mm}
	\caption{The approximate $k$-center of points computed using the $4$-approximation algorithm for $k$-center of points~\cite{malkomes2015fast}, for $k=20$ and $\epsilon=5$ (blue). The red points are the summary of~\Cref{alg:kclusterarbitrary} for $\epsilon=5$ and approximately represent the input shape. Note that the blue points are the centers of the disks that cover the red points, and the clustering radius is the maximum distance between a red point and its nearest blue point.}
	\label{fig:grid_gonz}\vspace*{-2mm}
\end{figure}

\eject
In the figures, only the grid points that lie inside the convex polygons are shown (in red). The computed centers are shown in blue.
 See \Cref{table:results} for the summary of the results.

\renewcommand{\arraystretch}{1.4}
\begin{table}[h!]
	\centering
	\caption{The summary of the experimental results.}
	\label{table:results}
	\begin{tabular}{|l|l|l|l|l|l|}
		\hline
		$k$ & $\epsilon$ & data-set & alg & summary size & radius\\
		\hline
		\hline
		\multirow{ 4}{*}{20} & - & input & \cite{malkomes2015fast} & 580,327 & $r=49.7757$ \\
		& $\epsilon=5$ & input & \Cref{alg:kclusterarbitrary} & 135,890 & $r=51.76$\\
		\cline{2-6}
		& - & test & \cite{malkomes2015fast} & 580,327 & $r= 65.1846$ \\
		& $\epsilon=5$ & test & \Cref{alg:kclusterarbitrary} & 135,890 & $r=51.76$\\
		\hline
	\end{tabular}
\end{table}

In our experiments, the radius of covering the input points are almost the same, which implies the effective sampling of our algorithm in practice. However, the set of vertices of a grid with cell-length $\epsilon$ inside each convex polygon has also been tested for the centers computed by the algorithm of \cite{malkomes2015fast}, which require a slightly larger radius than the centers computed with our algorithm.

\section{Conclusions and open problems}

In this paper, we introduced a new clustering problem so-called MinMax segments clustering.
Our algorithms for solving this problem are mostly bicriteria approximations.
Extending our model to other shapes such as polygons and disks remains open, as well as improving the approximation factors of our algorithms.
Solving these problems in the presence of outliers can also be interesting.

Also, we introduced the multi-interval set cover and proved it is NP-hard, and gave approximation algorithms for it, which can be of independent interest.

We generalized the region-based uncertainty model to allow multiple points from each region, to allow multiple centers to be chosen from the same region.
Assuming the input points are in the region-based uncertainty model with polygonal regions, we gave bicriteria approximation algorithms for the domain-restricted $k$-center of polygons in this model.

Further directions of research about aggregated uncertainty models can be adding weights to points, that are useful in problems where the number of points matters, such as k-means and capacitated clustering.

\subsection*{Acknowledgement}
Vahideh Keikha was supported by the Czech Science Foundation, grant number GJ19-06792Y and with institutional support RVO:67985807.


\end{document}